\documentclass{article}
\usepackage{ijcai26}

\usepackage{times}
\usepackage{soul}
\usepackage{url}
\usepackage[hidelinks]{hyperref}
\usepackage[utf8]{inputenc}
\usepackage[small]{caption}
\usepackage{graphicx}
\usepackage{amsmath}
\usepackage{amsthm}
\usepackage{booktabs}
\usepackage{algorithm}
\usepackage{algorithmic}
\usepackage[switch]{lineno}
\usepackage{dsfont}
\usepackage{comment}

\usepackage{amsfonts,amssymb}
\usepackage{balance} 
\usepackage{bm} 
\usepackage{pgfplots} 

\newcommand{\axisWidth}{8cm}
\newcommand{\axisHeight}{6cm}
\newcommand{\legendFont}{\small}

\usepackage{mathrsfs}
\usepackage[mathscr]{eucal}   
\usepackage{mathalfa}         
\usepackage{ifthen}

\newcommand{\vect}[1]{\vec{\mathbf{#1}}}
\newcommand{\vzero}{\vect{0}}

\newcommand{\vn}{\vect{n}}

\newcommand{\vx}{\vect{x}}
\newcommand{\vy}{\vect{y}}
\newcommand{\vN}{\vect{N}}
\newcommand{\vX}{\vect{X}}
\newcommand{\vY}{\vect{Y}}

\newcommand{\mA}{\mathcal{A}}

\newcommand{\mJ}{\mathcal{J}}
\newcommand{\mP}{\mathcal{P}}
\newcommand{\mR}{\mathcal{R}}
\newcommand{\rankings}{\mR}
\newcommand{\rankingsprime}{\mR'}
\newcommand{\rankingswith}[1]{\mR_{#1}}
\newcommand{\rankingsprimewith}[1]{\mR'_{#1}}

\newcommand{\proba}{\mathbb{P}}
\newcommand{\mean}{\mathbb{E}}

\newcommand{\bigO}{\mathcal{O}}
\newcommand{\Reals}{\mathbb{R}}

\newcommand{\sets}{S}
\newcommand{\sett}{T}

\newcommand{\SCW}{\operatorname{SCW}}

\newcommand{\firstsetofmanipulators}{\mathcal{V}_1}
\newcommand{\secondsetofmanipulators}{\mathcal{V}_2}

\newcommand{\RestateLemmaNumber}[1]{%
  \setcounter{lemma}{\numexpr #1-1\relax}%
}
\newcommand{\RestateTheoremNumber}[1]{%
  \setcounter{theorem}{\numexpr #1-1\relax}%
}

\newcommand{\pluralityWithRunoffTheoremRef}{1}
\newcommand{\arithmeticalLemmaRef}{1}

\newcommand{\pluralityWithRunoffTheorem}{%
For $m \geq 4$, the limit probability that Plurality with Runoff (PR) is susceptible to coalitional manipulation is equal to~$1$, that is,
\[
\proba_{\infty,m}(\textnormal{PR is CM}) = 1.
\]%
}

\newcommand{\arithmeticalLemma}{%
Let $\ell,n$ be positive integers and let $p$ be a nonnegative integer such that
\[
    p \;\leq\; \frac{n - \ell}{\ell + 1}.
\]
Then there exist integers $q_1, \ldots, q_\ell > p$ satisfying
\[
    q_1 + \cdots + q_\ell + p \;=\; n.
\]
There exists an integer $k$ with $0 \le k < \ell$ such that one solution is obtained by setting
$q_j = \left\lceil \frac{n - p}{\ell} \right\rceil - 1$ for $1 \le j \le k$, and
$q_j = \left\lceil \frac{n - p}{\ell} \right\rceil$ for $k < j \le \ell$.
}

\newboolean{isIJCAI}
\setboolean{isIJCAI}{false} 
\newcommand{\showIJCAIorHAL}[2]{%
  \ifthenelse{\boolean{isIJCAI}}{#1}{#2}%
}

\newtheorem{theorem}{Theorem}
\newtheorem{lemma}{Lemma}
\newtheorem{corollary}{Corollary}

\title{Super Condorcet Winners\\and Limit Coalitional Manipulability of IRV}

\author{
Élie de Panafieu$^1$
\and
François Durand$^1$\And
Guillem Perarnau$^2$\\
\affiliations
$^1$Nokia Bell Labs France\\
$^2$Universitat Politècnica de Catalunya\\
\emails
\{elie.de\_panafieu, francois.durand\}@nokia.com,
guillem.perarnau@upc.edu
}

\begin{document}

\maketitle

\begin{abstract}
We study the \emph{limit CM rate} of single-winner voting rules under Impartial Culture, defined as the probability that a preference profile is coalitionally manipulable in the limit of large electorates.
For $m=3$ candidates, Lepelley and Valognes [1999]
derived a closed-form expression for Plurality with Runoff, or equivalently Instant-Runoff Voting (IRV), and showed that its limit CM rate is strictly below~1.
This is remarkable because Kim and Roush [1996]
established a limit of~1 for several major rules, including Maximin and all positional scoring rules except Veto.
In this paper, we generalize the result of Lepelley and Valognes to any number of candidates $m \geq 4$.
We show that Plurality with Runoff has a limit CM rate equal to~1 for all $m \geq 4$, whereas IRV retains a limit CM rate strictly below~1.
To this end, we rely on the notion of \emph{Super Condorcet Winner}, recently introduced by Durand [2025],
which yields an upper bound on the CM rate of IRV.
We prove that this bound is asymptotically tight and compute the probability that a Super Condorcet Winner exists, thereby obtaining the exact limit CM rate of IRV.
\end{abstract}

\showIJCAIorHAL{%
An \href{https://shs.hal.science/hal-05566713}{\textcolor{blue}{extended version with appendices}}, the \href{https://github.com/francois-durand/limit_cm_rate_of_irv_ijcai_2026}{\textcolor{blue}{companion code}}, and a \href{https://youtu.be/4i2A7qUP-eo}{\textcolor{blue}{short video presentation}} are available online \cite{hal_version,code,video}.%
}{%
The \href{https://github.com/francois-durand/limit_cm_rate_of_irv_ijcai_2026}{\textcolor{blue}{companion code}} and a \href{https://youtu.be/4i2A7qUP-eo}{\textcolor{blue}{short video presentation}} are available online \cite{code,video}.%
}

\section{Introduction}

\subsection{Motivation}

\emph{Coalitional manipulation} (CM) occurs when a group of voters misreports their preferences to ensure the election of a candidate they all strictly prefer to the sincere winner.
By the Gibbard–Satterthwaite theorem \cite{gibbard1973manipulation,satterthwaite1975strategyproofness}, every non-trivial voting rule is vulnerable to such manipulation, even by a single voter.
The severity of this vulnerability is commonly measured by the \emph{CM rate}, defined either as the probability that a rule is coalitionally manipulable under a given random preference model, or as the proportion of preference profiles exhibiting coalitional manipulability in a real-world dataset.

A classical reference model is \emph{Impartial Culture} (IC), in which each voter's preference ranking is drawn independently and uniformly at random.
Although unrealistic, this model remains of considerable interest because it is widely regarded as a worst-case scenario for various voting properties \cite[Figure 7.6]{tsetlin2003impartialculture,durand2025irv,durand2015towards}.

In a pioneering paper, Kim and Roush \shortcite{kim1996manipulability} studied the asymptotic CM rate of several voting rules under Impartial Culture in the limit of large electorates.
They showed that for several prominent rules, the limit is one: Maximin, all positional scoring rules except Veto, and, for three candidates, Coombs' rule.
They also identified two exceptions: Veto and the rule that now bears their name.
Later, Lepelley and Valognes \shortcite{lepelley1999kimroush} identified another exception: for three candidates, Plurality with Runoff has a limit CM rate strictly below one.
In this case, it is equivalent to Instant-Runoff Voting (IRV), which proceeds by iteratively eliminating the candidate with the lowest plurality score.
Taken together, these results highlight the scarcity of voting rules whose limit CM rate under Impartial Culture is strictly below one.

IRV and Plurality with Runoff are of particular interest, as they are used in public elections in several countries.
Moreover, prior work identifies IRV as one of the most resilient voting rules with respect to coalitional manipulation (see Section~\ref{sec:related_work}).



In a recent paper, Durand \shortcite{durand2025irv} introduced the notion of \emph{Super Condorcet Winner}, a candidate whose plurality score exceeds the average when confronted with any subset of opponents, and showed that the existence of a Super Condorcet Winner suffices to make IRV immune to coalitional manipulation.
In other words, the proportion of profiles without a Super Condorcet Winner provides an upper bound on the CM rate of IRV, which was observed to be relatively tight in empirical datasets.
This makes the concept of Super Condorcet Winner a promising tool for analyzing the resilience of IRV to coalitional manipulation.

\subsection{Contributions}

We generalize the result of Lepelley and Valognes~\shortcite{lepelley1999kimroush} on Plurality with Runoff and IRV from three candidates to any number of candidates $m \geq 4$.
We first show that the limit CM rate of Plurality with Runoff is equal to~1, and then turn to the study of IRV.

We show that the existence of a Super Condorcet Winner is not only sufficient for IRV to be immune to coalitional manipulation, as established by Durand~\shortcite{durand2025irv}, but that under Impartial Culture the absence of a Super Condorcet Winner implies coalitional manipulability of IRV 
with probability tending to one as the number of voters tends to infinity.

This asymptotic equivalence allows us to characterize the limiting coalitional manipulability of IRV.
We prove that the probability that a Super Condorcet Winner exists under Impartial Culture converges to a positive limit and derive a closed-form expression for this probability, which directly yields the limiting CM rate of IRV and shows that it is strictly below~1.
Finally, we compute this limit numerically for up to fourteen candidates and perform a sanity check of our theoretical findings through large-scale Monte Carlo simulations.

\subsection{Related Work}\label{sec:related_work}

Beyond the references already discussed, a large body of work has examined coalitional manipulation.
IRV consistently emerges as one of the most resilient voting rules, whether across empirical studies on real preference profiles \cite{chamberlin1984observed,tideman2006collective,green2016statistical,durand2023coalitional} \cite[Chapter~9]{durand2015towards}, simulations under random preference models \cite{green2011four,green2014strategic,green2016statistical} \cite[Chapters~7--8]{durand2015towards}, and theoretical analyses \cite{lepelley1994vulnerability,lepelley1999kimroush,lepelley2003homogeneity,durand2025irv}.
The results of the present paper therefore add further theoretical support to this already substantial body of evidence in favor of IRV.

It is well known that IRV is not only infrequently susceptible to coalitional manipulation, but also computationally hard to manipulate: it is NP-complete to decide whether IRV is manipulable, whether by an individual or by a coalition \cite{bartholdi1991stv}.
The best current exact algorithm for coalitional manipulation runs in time $\bigO(m! \, n)$~\cite{coleman2007complexity}, where $m$ is the number of candidates and $n$ the number of voters.
In practice, experimental evidence \cite{vanderstraeten2010strategic} suggests that the computational burden of constructing a profitable strategic ballot under IRV can deter strategic attempts.

Another important recent contribution on coalitional manipulability is due to Xia \shortcite{xia2023impact}.
In his framework (see, e.g., Corollary~1), the maximal number of manipulators, called the \emph{budget} and denoted $B$, is fixed independently of the number of voters $n$, and the probability of coalitional manipulability essentially decreases as $\frac{B}{\sqrt{n}}$.
In contrast, we allow an arbitrary number of manipulators, provided they all strictly prefer the winner after manipulation.

\subsection{Limitations}

One may question the realism of coalitional manipulation itself, as large-scale coordination among voters may be difficult to achieve.
As discussed by Durand~\shortcite[Introduction]{durand2015towards} and more recently by Eggers and Nowacki~\shortcite{eggers2024susceptibility}, coalitional manipulation is best interpreted not as a model of \emph{ex ante} vulnerability to strategic voting, especially in large electorates, but rather as a measure of the \emph{ex post} regret experienced by a population after casting sincere ballots.

\subsection{Roadmap}

Section~\ref{sec:framework} introduces our framework.
Section~\ref{sec:plurality_with_runoff} studies Plurality with Runoff.
Section~\ref{sec:asymptotic_equivalence_scw_irv} shows that, with probability tending to~$1$ for large electorates, the existence of a Super Condorcet Winner is equivalent to IRV being immune to coalitional manipulation.
Section~\ref{sec:limit_probabilities} derives the limiting probability that a Super Condorcet Winner exists and applies this result to IRV.
Section~\ref{sec:future_work} concludes with directions for future research.


\section{Framework}\label{sec:framework}

Let the set of voters be denoted by $[n] = \{1, \ldots, n\}$ for a positive integer~$n$, and the set of candidates by $[m] = \{1, \ldots, m\}$ for an integer $m \ge 2$.
We use $a, b, c, d$ as typical variables for candidates.
Throughout the paper, $m$ is fixed and we study the asymptotic regime where the number of voters $n$ tends to infinity.
For any set $S$, we denote by $\mP(S)$ its power set and by $\mP^*(S)$ the set of its nonempty subsets.

We denote by $\rankings$ the set of all rankings over $[m]$, leaving the dependence on $m$ implicit.
A ranking $r \in \rankings$ may be written explicitly as $r = (1\!\succ\!2\!\succ\!3\!\succ\!4)$ for instance.
For any pair of candidates $c, b \in [m]$, let $\rankingswith{c \succ b}$ denote the subset of rankings that place $c$ above $b$.
More generally, for any subset of opponents $B \in \mP([m] \setminus \{c\})$, let $\rankingswith{c \succ B}$ denote the set of rankings in which $c$ is ranked above all candidates in~$B$.
Finally, for any subset $S \in \mP([m])$ (typically when $c \in S$), let $\rankingswith{c \succeq S}$ denote the set of rankings in which $c$ is ranked above every other candidate in~$S$.

A \emph{profile} is represented by a vector $\vn = (n_r)_{r \in \rankings}$ such that $\sum_{r \in \rankings} n_r = n$, where $n_r$ denotes the number of voters with preference ranking $r$ over the candidates.
Given a profile $\vn$, we define
\[
n_{c \succ b} = \!\!\!\sum_{r \in \rankingswith{c \succ b}}\!\! n_r, \quad
n_{c \succ B} = \!\!\!\sum_{r \in \rankingswith{c \succ B}}\!\! n_r, \quad
n_{c \succeq S} = \!\!\!\sum_{r \in \rankingswith{c \succeq S}}\!\! n_r,
\]
which respectively count the voters who rank candidate~$c$ above candidate~$b$, above all candidates in~$B$, and above all other candidates in~$S$.

A \emph{voting rule} is a function $f$ that maps each profile to a candidate in the set~$[m]$.
In this paper, we focus on two rules: Instant-Runoff Voting and Plurality with Runoff.

In \emph{Instant-Runoff Voting} (IRV), the winner is determined through successive eliminations.
The process starts with the full set $S_1 = [m]$.
At each round~$j$, IRV eliminates a candidate among those minimizing $n_{c \succeq S_j}$, referred to as the \emph{plurality score} of $c$ in~$S_j$, and the remaining candidates form the next set~$S_{j+1}$.
The procedure repeats until a single candidate remains, who is declared the winner.
We say that IRV \emph{visits} a set~$S$ if $S = S_j$ for some round~$j$ of the elimination process.

In \emph{Plurality with Runoff} (PR), the two candidates with the highest plurality scores in the full set of candidates advance to a second round.
The winner is the candidate preferred by a majority of voters in their pairwise comparison.

In both rules, ties are resolved using a predetermined tie-breaking rule; our theoretical results are independent of the particular tie-breaking rule used.

We say that a voting rule $f$ is \emph{coalitionally manipulable} (CM) in a profile $\vn$ (or equivalently that the profile~$\vn$ is CM in $f$) if there exists a \emph{target profile} $\vn' \neq \vn$ such that, for every ranking~$r$, whenever $n'_r < n_r$, the outcome $f(\vn')$ is preferred to $f(\vn)$ according to~$r$.
In other words, only voters who prefer the new outcome may have changed their ballots.
We use ``coalitionally manipulable'' interchangeably with ``susceptible to coalitional manipulation'', and likewise ``non coalitionally manipulable'' with ``immune to coalitional manipulation''.

A candidate $c$ is a \emph{Super Condorcet Winner} (SCW) in a profile~$\vn$ if, in every subset $S$ of candidates containing~$c$ and at least one other candidate, the plurality score of $c$ exceeds the average.
Formally, for every $S \in \mP([m])$ such that $c \in S$ and $|S| \geq 2$,
\begin{equation}\label{eq:scw_definition}
n_{c \succeq S} > \frac{n}{|S|}.
\end{equation}
The standard notion of a \emph{Condorcet winner} corresponds to satisfying this condition for all subsets $S$ of size~2.

Whenever an SCW exists, it is straightforward to verify that this candidate wins under IRV.
Moreover, in such profiles, IRV is immune to coalitional manipulation \cite{durand2025irv}.
Indeed, if IRV is coalitionally manipulable from a profile~$\vn$, with winner~$c$, to a target profile~$\vn'$, then applying IRV to~$\vn'$ must visit a \emph{violating subset} for~$\vn$, that is, a subset of candidates~$S$ for which Equation~\eqref{eq:scw_definition} does not hold.

Under \emph{Impartial Culture} (IC), each voter’s preference ranking is drawn independently and uniformly at random from the set~$\rankings$.
Formally, for each voter $v \in [n]$ and ranking $r \in \rankings$, let $X_r^{(v)}$ be the indicator variable equal to~$1$ if voter~$v$ draws~$r$, and~$0$ otherwise.
The Boolean vector $\vX^{(v)} = (X_r^{(v)})_{r \in \rankings}$ thus encodes the preference ranking of voter~$v$ and follows a multinomial distribution with one trial and uniform probabilities:
\(
\vX^{(v)} \sim \mathrm{MBin}\bigl(1; 1/m!, \ldots, 1/m!\bigr).
\)
The random profile under Impartial Culture is given by
\(
\vN = \sum_{v \in [n]} \vX^{(v)},
\)
which counts, for each ranking $r$, the number of voters with ranking~$r$.
By independence of the voters, $\vN$ follows a multinomial distribution with $n$ trials:
\(
\vN \sim \mathrm{MBin}\bigl(n; 1/m!, \ldots, 1/m!\bigr).
\)
We denote by $N_r$, $N_{c \succ b}$, $N_{c \succ B}$, and $N_{c \succeq S}$ the corresponding random variables, defined analogously to their deterministic counterparts.


We denote by $\proba_{n,m}$ the probability under Impartial Culture with $n$ voters and $m$ candidates, and by $\proba_{\infty,m}$ its limiting probability as $n \to \infty$, when it exists.

\section{Plurality with Runoff}\label{sec:plurality_with_runoff}

Before turning to our main results on IRV, we note that the result of Lepelley and Valognes \shortcite{lepelley1999kimroush}, which shows that Plurality with Runoff has a limit CM rate strictly below~1 for three candidates, does not extend to the case $m \geq 4$.

\begin{theorem}
\pluralityWithRunoffTheorem  
\end{theorem}

We provide only a proof sketch. The full proof is deferred to \showIJCAIorHAL{the technical appendix \cite{hal_version}}{Appendix~\ref{sec:appendix_pr}}.

\begin{proof}[Sketch of proof]
Fix $m \geq 4$, $\varepsilon > 0$, and $n > 0$, and consider a profile~$\vn$ whose normalized vector $\vn / n$ lies within $L_1$-distance at most~$\varepsilon$ from the uniform profile $(1/m!,\ldots,1/m!)$.
Let $c$ denote the winner of Plurality with Runoff at~$\vn$.
We show that, for $\varepsilon$ sufficiently small and $n$ large enough, the profile~$\vn$ is coalitionally manipulable in favor of any other candidate.

Let $a$ and $b$ be two distinct candidates, both different from~$c$.
For $\alpha \in [0,1]$, define a target profile~$\vn'$ as follows: among voters who prefer~$a$ to~$c$, a proportion~$\alpha$ report~$a$ as their top choice, while the remaining proportion~$1-\alpha$ report~$b$ as their top choice; all other ballots remain unchanged.
For $n$ sufficiently large, these proportions can be approximated arbitrarily closely.

Elementary calculations show that~$a$ and~$b$ advance to the runoff and that~$a$ wins the runoff provided that
\[
\frac{2}{3} < \alpha < \frac{m-1}{m},
\]
that $\varepsilon$ is sufficiently small, and that $n$ is large enough.
The condition on $\alpha$ is satisfiable for all $m \geq 4$.

Therefore, for $n$ large enough, any profile sufficiently close to the uniform profile is coalitionally manipulable under Plurality with Runoff.
By the weak law of large numbers, the probability of such profiles tends to~$1$ as $n \to \infty$.
\end{proof}

The remainder of the paper is devoted to extending the result of Lepelley and Valognes \shortcite{lepelley1999kimroush} to an arbitrary number of candidates, when interpreted as a result on IRV.

\section{Limit Equivalence Between SCW and IRV Immunity to Coalitional Manipulation} \label{sec:asymptotic_equivalence_scw_irv}

The existence of a Super Condorcet Winner is a sufficient condition for IRV to be immune to coalitional manipulation.
This condition is not necessary in general: there exist profiles without a Super Condorcet Winner for which IRV nevertheless remains immune to coalitional manipulation.
However, Durand~\shortcite{durand2025irv} observed that such counterexamples appear to be rare in real-world datasets.

In this section, we show that the same phenomenon occurs under Impartial Culture in the limit of large electorates: the probability of such counterexamples vanishes as the number of voters tends to infinity.
In other words, with high probability, IRV is immune to coalitional manipulation if and only if a Super Condorcet Winner exists.
To this end, we show that whenever no Super Condorcet Winner exists, and under mild conditions that are typically satisfied under Impartial Culture with large electorates, IRV is susceptible to coalitional manipulation.

\subsection{Strong Violation of the SCW Condition}

To build intuition, consider the problem of proving that IRV is susceptible to coalitional manipulation, and let $c$ denote the winner in the initial profile.
Any such manipulation must involve a subset $S = \{b_1, \ldots, b_\ell, c\}$ that leads to the elimination of~$c$ in the target profile.
Since manipulators cannot decrease the score of~$c$ below its initial value, denoted $p = n_{c \succeq S}$, they must instead ensure that the other candidates in~$S$ attain scores $q_1, \ldots, q_\ell$ at least equal to~$p$.
Moreover, to make the argument independent of the tie-breaking rule, we begin by considering situations in which these scores can be ensured to be strictly larger than~$p$.
The following arithmetic lemma provides a sufficient condition for the existence of such integers $q_1,\ldots,q_\ell$; its proof is elementary and deferred to \showIJCAIorHAL{the technical appendix \cite{hal_version}}{Appendix~\ref{sec:appendix_arithmetic}}.

\begin{lemma}  \label{th_average_arithmetic}
\arithmeticalLemma  
In addition, every such $q_j$ satisfies
\begin{equation}\label{eq:q_j_upper_bound}
    q_j \;<\; p \;+\; 2 \left(\frac{n}{\ell + 1} - p\right) + 1.
\end{equation}%
\end{lemma}

Lemma~\ref{th_average_arithmetic} provides a condition under which the opponents’ scores within a violating subset can, in principle, be rebalanced so that they all exceed that of the original winner.
The proposed solution equalizes these scores as much as possible, up to a difference of one due to integrality.
Equation \eqref{eq:q_j_upper_bound} will be used in the proof of Lemma~\ref{lem:main} to bound the number of ballot changes required for the manipulators.

Motivated by Lemma~\ref{th_average_arithmetic}, we say that a candidate~$c$ \emph{strongly violates the SCW condition} if there exists a subset $S$ of candidates containing~$c$ and at least one opponent such that
\begin{equation}\label{eq:strong_violation_definition}
    n_{c \succeq S} \;\leq\; \frac{n - (|S|-1)}{|S|}.
\end{equation}
This inequality is a direct reformulation of the condition in Lemma~\ref{th_average_arithmetic} under the change of variables $|S| = \ell + 1$.
Whenever it holds, we have in particular $n_{c \succeq S} \leq \frac{n}{|S|}$, so that $S$ is a violating subset and $c$ is not a Super Condorcet Winner.

\subsection{Manipulation Lemma}\label{sec:large_profiles_close_to_mean}

In addition to a strong violation of the SCW condition, the following lemma assumes that the profile is sufficiently large and sufficiently close to the expected profile. Under these assumptions, we prove that IRV is susceptible to coalitional manipulation.

\begin{lemma}\label{lem:main}
For every integer $m$ and every $\delta \in (0, 1)$,
there exists a constant $C_{m, \delta}$ such that any profile with $n > C_{m, \delta}$ voters and $m$ candidates satisfying
\begin{equation} \label{eq_close_to_mean}
    \left| n_r - \frac{n}{m!} \right| < n^{\delta}
    \quad \text{for all } r \in \rankings,
\end{equation}
and where the IRV winner strongly violates the SCW condition, is coalitionally manipulable in IRV, regardless of the tie-breaking rule.
\end{lemma}

\begin{proof}
Let $c$ denote the IRV winner.
Since the SCW condition is strongly violated, there exists a nonempty subset $B = \{b_1, \dots, b_\ell\}$ of other candidates such that
\[
n_{c \succ B} \le \frac{n - \ell}{\ell + 1}.
\]
Let $A = [m] \setminus (B \cup \{c\}) = \{a_1, \ldots, a_k\}$, with $k + \ell = m - 1$.
We construct a manipulation such that, in the modified profile, the IRV elimination order starts with $(a_1, \ldots, a_k, c)$.
All manipulators we consider rank $c$ sincerely last; hence, regardless of the final outcome of IRV, they all prefer the modified outcome to the original one.
We partition the manipulators into two disjoint groups, $\firstsetofmanipulators$ and $\secondsetofmanipulators$.
Voters in $\firstsetofmanipulators$ enforce the first $k$ eliminations $(a_1, \ldots, a_k)$, while voters in $\secondsetofmanipulators$ ensure that the next eliminated candidate is $c$.

\textbf{Concentration of plurality scores.}
Let $S$ be a nonempty subset of candidates and let $d \in S$.
The number of rankings in $\rankings$ satisfying $d \succeq S$ is $m! / |S|$.
Using the triangle inequality together with condition~\eqref{eq_close_to_mean}, we obtain
\begin{equation}
    \label{eq_useful_property}
    \left| n_{d \succeq S} - \frac{n}{|S|} \right|
=
    \bigg| \sum_{\substack{r \in \rankingswith{d \succeq S} }} \left( n_r - \frac{n}{m!} \right) \bigg|
<
    m! \, n^{\delta}.
\end{equation}

\textbf{Transformation 1: Eliminating candidates $a_1,\ldots,a_k$.}
By condition~\eqref{eq_close_to_mean}, the number of voters whose sincere ranking is
$(a_1\!\succ\!\cdots\!\succ\!a_k\!\succ\!b_1\!\succ\!\cdots\!\succ\!b_{\ell}\!\succ\!c)$
is greater than $\frac{n}{m!} - n^{\delta}$.
We select $2 \lceil m! \, n^{\delta} \rceil$ such voters to form the set $\firstsetofmanipulators$.
Since $\delta < 1$, taking $n$ large enough guarantees that
$2 \lceil m! \, n^{\delta} \rceil < \frac{n}{m!} - n^{\delta}$.
Voters in $\firstsetofmanipulators$ change their ballot to
$(b_1\!\succ\!\cdots\!\succ\!b_{\ell}\!\succ\!a_1\!\succ\!\cdots\!\succ\!a_k\!\succ\!c)$.

We now show that, after applying only Transformation~1, the IRV elimination order in the resulting profile~$\vn'$ starts with $(a_1,\ldots,a_k)$.
For $1\le j\le k$, suppose that after $j-1$ rounds the remaining candidates are
$S_j = \{a_j, \dots, a_k, b_1, \dots, b_{\ell}, c\}$.
By Equation~\eqref{eq_useful_property}, for any $d\in S_j\setminus\{a_j\}$,
\[
    n'_{d \succeq S_j} \geq n_{d \succeq S_j} > \frac{n}{m - (j - 1)} - m! \, n^{\delta},
\]
whereas
\[
    n'_{a_j \succeq S_j} = n_{a_j \succeq S_j} - 2 \lceil m! \, n^{\delta} \rceil
    < \frac{n}{m - (j - 1)} - m! \, n^{\delta}.
\]
Therefore, $a_j$ is eliminated at round~$j$.

\textbf{Transformation 2: Eliminating the original winner $c$.}
In the profile $\vn'$ obtained after Transformation~1, the set of remaining candidates at the beginning of round $k+1$ is $S_{k+1} = B \cup \{c\}$, which we hereafter simply denote by~$S$.
Moreover, since Transformation~1 does not modify the relative positions of candidates within $S$, the plurality scores $n'_{d \succeq S}$ for all $d\in S$ are the same as in the original profile.

Since candidate $c$ strongly violates the SCW condition on~$S$ in the original profile~$\vn$, we have
\[
    n'_{c \succeq S} = n_{c \succeq S} \leq \frac{n - \ell}{\ell + 1}.
\]
By Lemma~\ref{th_average_arithmetic}, there exist integers $q_1, \ldots, q_\ell > n_{c \succeq S}$ such that $q_1 + \cdots + q_\ell + n_{c \succeq S} = n$.
We will construct a manipulation leading to a profile~$\vn''$ such that $n''_{c \succeq S} = n_{c \succeq S}$ and $n''_{b_j \succeq S} = q_j$ for all $j\in [\ell]$.
This guarantees that $c$ is eliminated at round~$k+1$, independently of the tie-breaking rule.

To achieve the target scores on~$S$ without altering the elimination order established by Transformation~1, we select the manipulators in $\secondsetofmanipulators$ among voters whose sincere rankings start with
$(a_1\!\succ\!\cdots\!\succ\!a_k)$ and end with candidate~$c$.
We will modify the ballots of these voters only by permuting the candidates in~$B$, while keeping the positions of $a_1,\ldots,a_k$ and $c$ unchanged; this is possible if $\secondsetofmanipulators$ is disjoint from $\firstsetofmanipulators$, which we show feasible below.
Under such modifications, the plurality scores that determine the first $k$ elimination rounds are unaffected, so the eliminations of $(a_1,\ldots,a_k)$ still occur as in the profile~$\vn'$.
At the same time, this allows us to adjust the plurality scores $n''_{b_j \succeq S}$ within~$S$.

Let $\mJ$ denote the set of indices $j$ such that $n_{b_j \succeq S} > q_j$, corresponding to candidates with an excess of voters.
The set $\secondsetofmanipulators$ is constructed as the union, over all $j \in \mJ$, of $n_{b_j \succeq S} - q_j$ arbitrary voters whose sincere ranking starts with $(a_1\!\succ\!\cdots\!\succ\!a_k\!\succ\!b_j)$ and ends with $c$.
For each $i \in [\ell] \setminus \mJ$, corresponding to a candidate with no excess of voters, we ask $q_i - n_{b_i \succeq S}$ manipulators from $\secondsetofmanipulators$ to swap $b_i$ with the candidate currently in position $(k+1)$ in their ranking, thus casting a ballot that begins with $(a_1\!\succ\!\cdots\!\succ\!a_k\!\succ\!b_i)$.
This construction is feasible because the total excess equals the total deficit:
\[
    \sum_{j \in \mJ} (n_{b_j \succeq S} - q_j)
    - \sum_{i \in [\ell] \setminus \mJ} (q_i - n_{b_i \succeq S})
    = 0.
\]
After this reassignment, we have $n''_{b_j \succeq S} = q_j$ for all $j \in [\ell]$, which completes the construction.

\textbf{The sets $\firstsetofmanipulators$ and $\secondsetofmanipulators$ can be chosen disjoint.}
We have already seen that the set $\firstsetofmanipulators$ has size $\bigO(n^{\delta})$.
We now show that the same holds for $\secondsetofmanipulators$.
Since the number of voters with sincere ranking $(a_1\!\succ\!\cdots\!\succ\!a_k\!\succ\!b_1\!\succ\!\cdots\!\succ\!b_{\ell}\!\succ\!c)$ is at least $\frac{n}{m!} - n^{\delta}$, it follows that, for sufficiently large $n$, the set $\secondsetofmanipulators$ can be chosen disjoint from $\firstsetofmanipulators$.

Let us bound the size of $\secondsetofmanipulators$.
By Equation~\eqref{eq_useful_property}, for every $j\in[\ell]$,
\[
    n_{b_j \succeq S} = \frac{n}{\ell+1} + \bigO(n^{\delta}),
    \qquad
    n_{c \succeq S} = \frac{n}{\ell+1} + \bigO(n^{\delta}).
\]
The triangle inequality and $q_j > n_{c \succeq S}$ then imply
\[
    |q_j-n_{b_j \succeq S}|
    =
    \left|
    q_j - n_{c \succeq S} + \bigO(n^{\delta})
    \right|
    \leq
    q_j - n_{c \succeq S} + \bigO(n^{\delta}).
\]
Using the upper bound~\eqref{eq:q_j_upper_bound} from Lemma~\ref{th_average_arithmetic}, we deduce
\[
    |q_j - n_{b_j \succeq S}| = \bigO(n^{\delta})
    \quad \text{for all } j\in[\ell].
\]
Therefore, the size of $\secondsetofmanipulators$, given by $\sum_{j\in\mJ} |q_j - n_{b_j \succeq S}|$, is $\bigO(n^{\delta})$.
%
\end{proof}

\subsection{Limit Equivalence Theorem}\label{sec:limit_equivalence_theorem}


Before stating the theorem, we briefly provide some intuition.
Lemma~\ref{lem:main} applies to profiles satisfying three conditions:
\begin{itemize}
    \item the number of voters $n$ is sufficiently large;
    \item the deviations $n_r - n/m!$ from the expected profile are small relative to $n$; and
    \item the IRV winner strongly violates the SCW condition.
\end{itemize}
For such profiles, the lemma shows that IRV is coalitionally manipulable.
Under Impartial Culture with large electorates, the first two conditions are typically satisfied, and whenever the IRV winner violates the SCW condition, this violation is typically strong.
As a result, with high probability, the absence of an SCW suffices for IRV to be coalitionally manipulable.
The following theorem formalizes this intuition.

\begin{theorem}\label{thm:asymptotic_equivalence_scw_irv_not_cm}
Under Impartial Culture, with high probability as the number of voters tends to infinity, IRV is immune to coalitional manipulation if and only if a Super Condorcet Winner exists.
\end{theorem}

\begin{proof}
Fix an arbitrary $\delta \in (1/2, 1)$.
Define the event $\mA = \mA_1 \cup \mA_2$, where
\begin{align*}
    \mA_1&= \bigcup_{r \in \rankings}\left\{\, \left|N_r-\frac{n}{m!}\right|\geq n^\delta \,\right\},\\
    \mA_2&= \bigcup_{\substack{S\subseteq [m] \\ |S|\geq 2}} \bigcup_{c\in S}
    \left\{\,\frac{n - |S| + 1}{|S|} < N_{c \succeq S} \leq \frac{n}{|S|} \,\right\}.
\end{align*}
Intuitively, $\mA_1$ corresponds to profiles whose deviations from the mean are too large for Lemma~\ref{lem:main} to apply, while $\mA_2$ captures cases in which a candidate violates the SCW condition, but not strongly. We will show that $\proba_{n,m}(\mA) = o(1)$.

On the one hand, since each $N_r$ is a binomial random variable with $n$ trials and success probability $1/m!$, Chernoff's inequality \cite[Corollary~A.1.7]{alonspencer2015}, together with a union bound over all $r \in \rankings$, yields
\[
    \proba_{n,m}(\mA_1) \le 2 m! \exp(-2 n^{ 2\delta-1})  = o(1).
\]
On the other hand, each $N_{c \succeq S}$ is a binomial random variable with $n$ trials and success probability $1/|S|$.
Since the maximum point probability of a binomial distribution is $\bigO(1/\sqrt{n p (1-p)})$, a union bound over all subsets $S \subset [m]$ with $|S| \ge 2$ and all $c \in S$ gives
\[
    \proba_{n,m}(\mA_2)\leq \sum_{k=2}^m m \binom{m}{k} \bigO(1/\sqrt{n})=o(1).
\]
Therefore, $\proba_{n,m}(\mA) = o(1)$.

Since the existence of an SCW implies that IRV is immune to coalitional manipulation, it suffices to show that
\[
    \proba_{n, m}(\nexists \text{ SCW and IRV is not CM})
\]
tends to zero as $n \to \infty$.
By Lemma~\ref{lem:main}, this event is included in $\mA$, which concludes the proof.
\end{proof}



\section{Limit Probabilities}\label{sec:limit_probabilities}

Building on the results of Section~\ref{sec:asymptotic_equivalence_scw_irv}, the analysis of the limit probability that IRV is immune to coalitional manipulation reduces to studying the probability that a Super Condorcet Winner exists.
Section~\ref{sec:limit_scw_positive} proves that this probability is strictly positive by expressing it as a high-dimensional Gaussian integral.
Section~\ref{sec:limit_scw_expression} derives an alternative, lower-dimensional expression that remains computationally demanding but is more amenable to numerical computation for moderate numbers of candidates.
Both proofs are inspired by classical results on the probability of existence of a Condorcet Winner~\cite{niemi1968mathematical,krishnamoorthy2005condorcet}.
Finally, Section~\ref{sec:limit_irv_is_cm} uses these expressions to compute the limit CM rate of IRV and performs a sanity check by comparing the resulting values with prior findings and Monte Carlo simulations.

\subsection{Limit Probability of SCW: Positivity}\label{sec:limit_scw_positive}

\begin{theorem}\label{thm_limit_scw_positive}
For every integer $m \geq 2$, the limit probability that a Super Condorcet Winner exists is strictly positive:
\[
  \proba_{\infty, m}(\exists \SCW) > 0.
\]
\end{theorem}

Classical results on the existence of a Condorcet winner~\cite{niemi1968mathematical,krishnamoorthy2005condorcet} imply that, for $m \geq 3$, this limit is strictly below~$1$, and that it tends to~$0$ as $m \to \infty$.

The proof of Theorem~\ref{thm_limit_scw_positive} relies on the following lemma.

\begin{lemma}\label{thm_limit_scw_expensive_formula}
Let $\rankingsprime = \rankings \setminus \{(1\!\succ\!\cdots\!\succ\!m)\}$, and let
$K \subseteq \Reals^{m! - 1}$ be the cone of vectors
$\vx = (x_r)_{r \in \rankingsprime}$ such that, for every nonempty subset $\sets \in \mathcal{P}^*([m-1])$ of opponents of candidate~$m$,
\[
    \sum_{r \in \rankingswith{m \succ \sets}} x_r > 0.
\]
Let $M$ be the positive definite matrix defined by
\[
  M \;=\; \frac{1}{m!}\, I \;-\; \frac{1}{m!^2}\, J,
\]
where $I$ denotes the identity matrix and $J$ the all-ones matrix, both of dimension
$(m!-1) \times (m!-1)$.
Then the limit probability that a Super Condorcet Winner exists is given by
\begin{equation} \label{eq_limit_proba_factorial_m}
    \proba_{\infty, m}( \exists \SCW) =
        \frac{m}{\sqrt{(2 \pi)^{m!-1} \det(M)}}
    \int_K
    e^{- \vx^T M^{-1} \vx / 2}
    d \vx.
\end{equation}
\end{lemma}

\begin{proof}[Proof of Lemma~\ref{thm_limit_scw_expensive_formula}]
When a Super Condorcet Winner exists, it is unique.
By symmetry among candidates,
\[
    \proba_{n,m}(\exists \SCW)
    \;=\;
    m \cdot \proba_{n,m}(\textnormal{candidate $m$ is an SCW}) .
\]

Let $\rankingsprime = \rankings \setminus \{(1\!\succ\!\cdots\!\succ\!m)\}$, which removes one redundant coordinate due to the constraint $\sum_{r \in \rankings} N_r = n$.
By the multivariate central limit theorem, the normalized vector
\[
    \hat{\vN}
    =
    \left(
        \frac{N_r - n/m!}{\sqrt{n}}
    \right)_{r \in \rankingsprime}
\]
converges in distribution, as $n \to \infty$, to a centered multivariate normal
distribution with covariance matrix~$M$.

Candidate~$m$ is an SCW if and only if
\[
  \sum_{r \in \rankingswith{m \succ \sets}} N_r
  \;>\;
  \frac{n}{|\sets| + 1},
  \qquad
  \text{for all }\sets \in \mathcal{P}^*([m-1]).
\]
These inequalities do not involve the excluded coordinate
$N_{(1\succ\cdots\succ m)}$ and may therefore be restricted to rankings in
$\rankingsprime$.
After centering and normalization, these conditions become
\[
    \sum_{r \in \rankingsprimewith{m \succ \sets}} \hat{N}_r > 0
    \qquad
  \text{for all }\sets \in \mathcal{P}^*([m-1]).
\]
Therefore, in the limit, candidate~$m$ is the SCW if and only if the Gaussian vector
$\hat{\vN}$ lies in the cone~$K$.
\end{proof}

\begin{proof}[Proof of Theorem~\ref{thm_limit_scw_positive}]
Consider the vector $\hat{\vN}$ defined by $\hat{N}_r = 1$ if $r = (m \succ \cdots \succ 1)$, $\hat{N}_r = -1$ if $r = (1 \succ \cdots \succ m)$, and $\hat{N}_r = 0$ otherwise.
This vector strictly satisfies all inequalities defining~$K$ and therefore lies in the interior of~$K$.
Hence, $K$ has positive Lebesgue measure in $\Reals^{m!-1}$.
Since the limiting distribution in~\eqref{eq_limit_proba_factorial_m} admits a continuous density that is strictly positive everywhere, the corresponding probability is strictly positive.
\end{proof}

\subsection{Limit Probability of SCW: Expression}\label{sec:limit_scw_expression}

The expression~\eqref{eq_limit_proba_factorial_m} given in Lemma~\ref{thm_limit_scw_expensive_formula} is difficult to compute: the integral has dimension $m!-1$, and its domain $K$ is defined by a system of inequalities.
The next theorem addresses both issues by deriving an alternative representation in dimension
$d = 2^{m-1}-1$, which remains exponential in~$m$ but is significantly more tractable, and whose domain of integration is the positive orthant $(0,\infty)^d$.


\begin{theorem}
\label{th_simpler_expression_limit_proba}
Assume $m \geq 2$, and let $d = 2^{m-1} - 1$ be the cardinality of $\mP^*([m-1])$, that is, the number of nonempty subsets of opponents of candidate~$m$.
Let $H$ be the matrix indexed by $\mP^*([m-1])$, defined for any $\sets, \sett \in \mP^*([m-1])$ by
\[
    H_{\sets, \sett} \;=\;
    \frac{1}{|\sets \cup \sett| + 1}
    - \frac{1}{|\sets| + 1} \cdot \frac{1}{|\sett| + 1}.
\]
Then the limit probability that a Super Condorcet Winner exists is given by
\begin{equation} \label{eq_limit_proba_dim_subsets}
    \proba_{\infty, m}(\exists \SCW) \!=\!
    \frac{m}{\sqrt{(2 \pi)^d \det(H)}}
    \!\int_{(0, +\infty)^d}\!
    e^{- \vx^T H^{-1} \vx / 2}
    d \vx.
\end{equation}
\end{theorem}

\begin{proof}
Let $\sets, \sett \in \mP^*([m-1])$, and let $r$ be a ranking drawn uniformly at random, interpreted as the preference of a single voter.
Define the indicator random variable $Y_{m \succ \sets}$ to be $1$ if $r \in \mR_{m \succ \sets}$, and $0$ otherwise.
Denoting by $\mean[\cdot]$ the expectation, symmetry among the candidates in $\sets \cup \{m\}$ and in $\sets \cup \sett \cup \{m\}$ yields
\[
    \mean(Y_{m \succ \sets})\!=\!\frac{1}{|\sets| + 1}
    \;\;\text{and}\;\; \mean(Y_{m \succ \sets} \, Y_{m \succ \sett})\!=\!\frac{1}{|\sets \cup \sett| + 1}.
\]
Denoting by $H_{\sets, \sett}$ the covariance of $Y_{m \succ \sets}$ and $Y_{m \succ \sett}$, it follows that
\[
    H_{\sets, \sett}
    = \frac{1}{|\sets \cup \sett| + 1}
       - \frac{1}{|\sets| + 1} \cdot \frac{1}{|\sett| + 1}.
\]

Fix an arbitrary ordering of $\mP^*([m-1])$.
Let $\vY$ be the $d$-dimensional vector with entries $Y_{m \succ \sets}$, and let $\hat{\vY}$ be its centered version with entries $Y_{m \succ \sets} - \tfrac{1}{|\sets|+1}$.
Let $H$ denote the covariance matrix of $\vY$, with entries $H_{\sets,\sett}$.
As a covariance matrix, $H$ is positive semidefinite; we now show that it is in fact positive definite.

Consider a vector $\vy$ such that $\vy^T H \vy = 0$. Then
\[
  0
  =
  \vy^T \mean\!\left[\hat{\vY}\hat{\vY}^T\right] \vy
  =
  \mean\!\left[(\vy^T \hat{\vY})^2\right],
\]
which implies that $\vy^T \hat{\vY} = 0$ almost surely.
Equivalently, there exists a constant $C$ such that
$\vy^T \vY = \sum_\sets y_\sets Y_{m\succ \sets} = C$ almost surely.
Taking a ranking in which candidate~$m$ is ranked last yields $\vY = \vzero$, and therefore $C=0$.

Let $\sett \in \mP^*([m-1])$ with $|\sett| = k \geq 1$, and assume by induction that
$y_{\sets}=0$ for all $\sets \in \mP^*([m-1])$ such that $|\sets| < k$ (for $k=1$, this condition is void).
Consider a ranking whose last $k+1$ positions are
\(
(m \succ d_1 \succ \cdots \succ d_k),
\)
where $\{d_1,\ldots,d_k\}=\sett$.
For this ranking, we have $Y_{m \succ \sets}=1$ if and only if $\varnothing \neq \sets \subseteq \sett$. Since $\sum_{\sets} y_{\sets} Y_{m \succ \sets}=0$ almost surely, it follows that
\[
  0
  \;=\;
  \sum_{\sets} y_{\sets} Y_{m \succ \sets}
  \;=\;
  \sum_{\varnothing \neq \sets \subseteq \sett} y_{\sets}.
\]
By the induction hypothesis, all terms in the sum except $y_{\sett}$ vanish, hence $y_{\sett}=0$.

Since this holds for every $\sett \in \mP^*([m-1])$, we conclude that $\vy=\vzero$, and thus $H$ is positive definite.

Finally, consider a random profile drawn from Impartial Culture with $n$ voters and $m$ candidates.
For each $\sets \in \mP^*([m-1])$, we have
\[
  N_{m \succ \sets} = \sum_{v=1}^n X_{m \succ \sets}^{(v)},
  \qquad
  \hat{N}_{m \succ \sets} = \frac{N_{m \succ \sets} - n/(|\sets|+1)}{\sqrt{n}},
\]
where the random vectors $\bigl(X_{m \succ \sets}^{(v)}\bigr)_{\sets \in \mP^*([m-1])}$ are i.i.d. and follow the same distribution as~$\vY$.
By the multivariate central limit theorem, the vector $(\hat{N}_{m \succ \sets})_{\sets \in \mP^*([m-1])}$ converges in distribution to a centered Gaussian vector with covariance matrix~$H$.

The SCW condition for candidate~$m$ is that $N_{m \succ \sets}$ exceeds its expectation for all $\sets \in \mP^*([m-1])$, which is equivalently expressed as $\hat{N}_{m \succ \sets} > 0$ for all such~$\sets$.
Hence, the limit probability that candidate~$m$ is an SCW is given by the Gaussian measure of the positive orthant in $\Reals^d$.
Multiplying by~$m$ then yields the expression stated in Equation~\eqref{eq_limit_proba_dim_subsets}.
\end{proof}

\subsection{Limit Probability of IRV Susceptibility to Coalitional Manipulation}\label{sec:limit_irv_is_cm}

By Theorem~\ref{thm:asymptotic_equivalence_scw_irv_not_cm}, the limit probability that IRV is immune to coalitional manipulation coincides with the limit probability that a Super Condorcet Winner exists.
Theorems~\ref{thm_limit_scw_positive} and~\ref{th_simpler_expression_limit_proba} further provide explicit characterizations of this probability.
Together, these results yield the following corollary.

\begin{corollary}\label{thm_proba_irv_cm}
The limit probability that IRV is susceptible to coalitional manipulation satisfies the following properties:
\begin{enumerate}
  \item For every integer $m \geq 2$,
  \[
    \proba_{\infty,m}(\textnormal{IRV is CM}) < 1.
  \]

  \item With the notation of Theorem~\ref{th_simpler_expression_limit_proba},
  \[
  \begin{aligned}
    &\proba_{\infty, m}(\textnormal{IRV is CM}) \\
    &=
    1 - \frac{m}{\sqrt{(2 \pi)^d \det(H)}}
    \int_{(0, +\infty)^d}
    e^{- \vx^T H^{-1} \vx / 2}
    d \vx.
  \end{aligned}
  \]
\end{enumerate}
\end{corollary}

For $m \geq 3$, this limit is also strictly positive and tends to~$1$ as $m \to \infty$, by symmetry with the corresponding result for the existence of a Super Condorcet Winner, which itself follows from classical results on the Condorcet winner.


Note that Corollary~\ref{thm_proba_irv_cm} applies identically to \emph{Exhaustive Ballot}, the round-by-round implementation of IRV, whose CM rate is lower bounded by the one of IRV and upper bounded by the probability that no SCW exists~\cite{durand2025irv}.

Table~\ref{tab:numerical_values_limits_irv} reports the limit probability $\proba_{\infty, m}(\textnormal{IRV is CM})$ for several values of~$m$, obtained by numerically evaluating the integral expression of Corollary~\ref{thm_proba_irv_cm} (see the supplementary material for implementation details and documentation).
For $m=3$, we verify that the resulting value is consistent with the findings of Lepelley and Valognes~\shortcite{lepelley1999kimroush}.

%

\begin{table}[ht]
    \centering
\hfill\begin{minipage}{0.45\linewidth}
\centering
\begin{tabular}{r r}
\toprule
$m$ & $\proba_{\infty,m}(\textnormal{IRV is CM})$ \\
\midrule
3 & $0.1688 \pm 0.0013$ \\
4 & $0.3472 \pm 0.0015$ \\
5 & $0.4986 \pm 0.0015$ \\
6 & $0.6219 \pm 0.0015$ \\
7 & $0.7156 \pm 0.0014$ \\
8 & $0.7850 \pm 0.0013$ \\
\bottomrule
\end{tabular}
\end{minipage}
\hfill
\begin{minipage}{0.45\linewidth}
\centering
\begin{tabular}{r r}
\toprule
$m$ & $\proba_{\infty,m}(\textnormal{IRV is CM})$ \\
\midrule
9  & $0.8381 \pm 0.0012$ \\
10 & $0.8795 \pm 0.0011$ \\
11 & $0.9105 \pm 0.0010$ \\
12 & $0.9328 \pm 0.0009$ \\
13 & $0.9500 \pm 0.0008$ \\
14 & $0.9620 \pm 0.0007$ \\
\bottomrule
\end{tabular}
\end{minipage}\hfill
    \caption{Numerical estimates of the limit probability that IRV is susceptible to coalitional manipulation, as given by Corollary~\ref{thm_proba_irv_cm}. Each value is reported with an error bound reflecting the numerical integration.}
    \label{tab:numerical_values_limits_irv}
\end{table}

\begin{figure}[ht]
    \centering
\begin{tikzpicture}

\definecolor{crimson2143940}{RGB}{214,39,40}
\definecolor{darkgrey176}{RGB}{176,176,176}
\definecolor{darkorange25512714}{RGB}{255,127,14}
\definecolor{forestgreen4416044}{RGB}{44,160,44}
\definecolor{grey127}{RGB}{127,127,127}
\definecolor{lightgrey204}{RGB}{204,204,204}
\definecolor{mediumpurple148103189}{RGB}{148,103,189}
\definecolor{orchid227119194}{RGB}{227,119,194}
\definecolor{sienna1408675}{RGB}{140,86,75}
\definecolor{steelblue31119180}{RGB}{31,119,180}

\begin{axis}[reverse legend,
height=\axisHeight,
legend cell align={left},
legend style={font=\legendFont, 
  fill opacity=1,
  draw opacity=1,
  text opacity=1,
  at={(0.03,0.97)},
  anchor=north west,
  draw=lightgrey204
},
log basis x={10},
tick align=outside,
tick pos=left,
width=\axisWidth,
x grid style={darkgrey176},
xlabel={Number of voters $n$},
xmin=10, xmax=1000,
xmode=log,
xtick style={color=black},
xmode=log,
y grid style={darkgrey176},
ylabel={$\mathbb{P}_{n, m}(\text{IRV is CM})$},
ymajorgrids,
ymin=-0.05, ymax=1.05,
ytick={0.0, 0.1, 0.2, 0.30000000000000004, 0.4, 0.5, 0.6000000000000001, 0.7000000000000001, 0.8, 0.9, 1.0},
ytick style={color=black}
]
\addplot [semithick, steelblue31119180]
table {%
10 0.12024
30 0.14234
100 0.15423
300 0.16099
1000 0.1631
};
\addlegendentry{$m=3$}
\addplot [semithick, darkorange25512714]
table {%
10 0.24753
30 0.29642
100 0.32263
300 0.32966
1000 0.33734
};
\addlegendentry{$m=4$}
\addplot [semithick, forestgreen4416044]
table {%
10 0.36216
30 0.43004
100 0.46408
300 0.48026
1000 0.48627
};
\addlegendentry{$m=5$}
\addplot [semithick, crimson2143940]
table {%
10 0.45701
30 0.54239
100 0.58159
300 0.59582
1000 0.61005
};
\addlegendentry{$m=6$}
\addplot [semithick, mediumpurple148103189]
table {%
10 0.53702
30 0.63095
100 0.66954
300 0.68748
1000 0.70236
};
\addlegendentry{$m=7$}
\addplot [semithick, sienna1408675]
table {%
10 0.60398
30 0.69997
100 0.7442
300 0.75861
1000 0.77286
};
\addlegendentry{$m=8$}
\addplot [semithick, orchid227119194]
table {%
10 0.65643
30 0.75551
100 0.79684
300 0.81812
1000 0.82803
};
\addlegendentry{$m=9$}
\addplot [semithick, grey127]
table {%
10 0.6998
30 0.80206
100 0.84274
300 0.85829
1000 0.86833
};
\addlegendentry{$m=10$}
\path [draw=steelblue31119180, semithick, dash pattern=on 5.55pt off 2.4pt]
(axis cs:10,0.16882)
--(axis cs:1000,0.16882);

\path [draw=darkorange25512714, semithick, dash pattern=on 5.55pt off 2.4pt]
(axis cs:10,0.347156)
--(axis cs:1000,0.347156);

\path [draw=forestgreen4416044, semithick, dash pattern=on 5.55pt off 2.4pt]
(axis cs:10,0.498575)
--(axis cs:1000,0.498575);

\path [draw=crimson2143940, semithick, dash pattern=on 5.55pt off 2.4pt]
(axis cs:10,0.62185)
--(axis cs:1000,0.62185);

\path [draw=mediumpurple148103189, semithick, dash pattern=on 5.55pt off 2.4pt]
(axis cs:10,0.715569)
--(axis cs:1000,0.715569);

\path [draw=sienna1408675, semithick, dash pattern=on 5.55pt off 2.4pt]
(axis cs:10,0.785016)
--(axis cs:1000,0.785016);

\path [draw=orchid227119194, semithick, dash pattern=on 5.55pt off 2.4pt]
(axis cs:10,0.838144)
--(axis cs:1000,0.838144);

\path [draw=grey127, semithick, dash pattern=on 5.55pt off 2.4pt]
(axis cs:10,0.87949)
--(axis cs:1000,0.87949);

\end{axis}

\end{tikzpicture}
    \caption{Probability that IRV is susceptible to coalitional manipulation as a function of~$n$, for various values of~$m$. Each point is estimated by Monte Carlo simulation with $100{,}000$ samples, yielding an error of order $1/\sqrt{100{,}000}\approx 0.3\%$. Simulations rely on the Python package SVVAMP. The theoretical limits from Table~\ref{tab:numerical_values_limits_irv} are shown as dashed lines.}
    \label{fig:irv_cm_rates}
\end{figure}
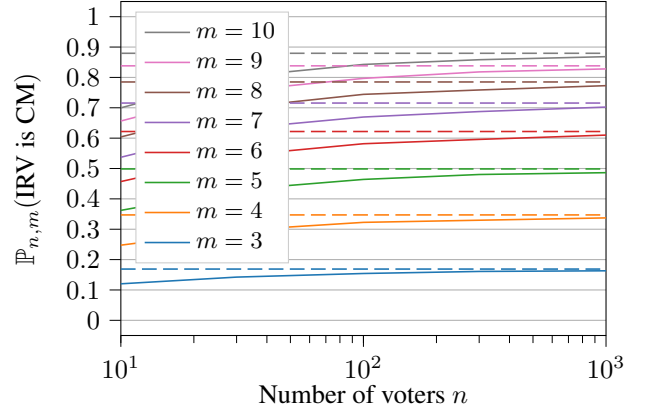

Figure~\ref{fig:irv_cm_rates} provides an additional sanity check by comparing our theoretical predictions with Monte Carlo experiments using SVVAMP~\cite{durand2016svvamp,svvamp}, a Python package dedicated to studying manipulability of voting rules.
The close agreement between the empirical curves and the predicted limits supports the validity of our results.


\section{Future Work}\label{sec:future_work}

A natural next step is to extend our analysis beyond IRV and determine the limiting CM rate under Impartial Culture for a broader class of classical voting rules, thereby filling the gaps left by Kim and Roush \shortcite{kim1996manipulability}.
Their work also introduces the notion of the smallest achievable CM rate among anonymous and neutral rules that are well-defined on a proportion of profiles tending to one.
It is therefore natural to ask whether IRV attains this theoretical lower bound, about which little is currently known.
If not, identifying voting rules that achieve it would be of particular interest, especially in light of their empirical performance on real-world preference datasets.

\appendix

%
%
\section*{Acknowledgments}

E.d.P. and G.P. acknowledge support from the RandNET project (MSCA-RISE — Marie Skłodowska-Curie Research and Innovation Staff Exchange Programme, Grant Agreement No. 101007705), which funded the research secondments during which this project was initiated and developed.

G.P. was supported by the project PID2023-147202NB-I00 funded by MICIU/AEI/10.13039/501100011033/.

\bibliographystyle{named}
\bibliography{main}

\showIJCAIorHAL{}{%
\clearpage

In this technical appendix, we provide the proofs of Theorem~{\pluralityWithRunoffTheoremRef} and Lemma~{\arithmeticalLemmaRef}. For the reader’s convenience, we restate these results before presenting their proofs.

\section{Plurality with Runoff}\label{sec:appendix_pr}

We first prove Theorem~{\pluralityWithRunoffTheoremRef}, establishing the limit CM rate of Plurality with Runoff for $m \ge 4$.

\RestateTheoremNumber{\pluralityWithRunoffTheoremRef}
\begin{theorem}
\pluralityWithRunoffTheorem    
\end{theorem}

\begin{proof}
Fix $m \geq 4$, $\varepsilon > 0$, and $n > 0$, and consider a profile~$\vn$ whose normalized vector $\vn / n$ lies within $L_1$-distance at most~$\varepsilon$ from the uniform profile $(1/m!,\ldots,1/m!)$.
Let $c$ denote the winner of Plurality with Runoff at~$\vn$.
We show that, for $\varepsilon$ sufficiently small and $n$ large enough, the profile~$\vn$ is coalitionally manipulable in favor of any other candidate.

Let $a$ and $b$ be two distinct candidates, both different from~$c$.
For $\alpha \in [0,1]$, we define a target profile~$\vn'$ as follows.
In the original profile~$\vn$, the number of voters who prefer~$a$ to~$c$ satisfies
\[
n_{a \succ c} = \frac{n}{2} \pm n \varepsilon
\]
Among these voters, we ask $\lfloor \alpha n_{a \succ c} \rfloor$ of them to report~$a$ as their first choice, and the remaining ones to report~$b$ as their first choice.
All other ballots remain unchanged.
In the resulting profile~$\vn'$, we have
\begin{align*}
n'_{a \succ [m]} &= \alpha \frac{n}{2} \pm n \varepsilon \pm 1, \\
n'_{b \succ [m]} &= \frac{n}{2m} + (1 - \alpha) \frac{n}{2} \pm n \varepsilon \pm 1, \\
n'_{c \succ [m]} &= \frac{n}{m} \pm n \varepsilon, \\
n'_{d \succ [m]} &= \frac{n}{2m} \pm n \varepsilon \quad \textrm{for all $d \notin \{a, b, c\}$}, \\
n'_{a \succ b} &= \frac{n}{6} + \alpha \frac{n}{2} \pm n \varepsilon \pm 1.
\end{align*}

Now assume that 
\[
\frac{2}{3} < \alpha < \frac{m-1}{m},
\]
which is feasible since $m \geq 4$.
For $\varepsilon$ sufficiently small and $n$ large enough, we obtain the following conclusions.
\begin{itemize}
\item Since $\alpha > \frac{2}{m}$, we have $n'_{a \succ [m]} > n'_{c \succ [m]}$.
\item Since $\alpha < \frac{m - 1}{m}$, we have $n'_{b \succ [m]} > n'_{c \succ [m]}$.
\item Moreover, $n'_{c \succ [m]} > n'_{d \succ [m]}$ for all $d \notin \{a,b,c\}$.
\item As a consequence, candidates $a$ and $b$ advance to the runoff.
\item Finally, since $\alpha > \frac{2}{3}$, we have $n'_{a \succ b} > \frac{n}{2}$, and candidate~$a$ wins the runoff.
\end{itemize}

Therefore, for $n$ large enough, any profile sufficiently close to the uniform profile is coalitionally manipulable under Plurality with Runoff.
By the weak law of large numbers, the probability of such profiles tends to~$1$ as $n \to \infty$.
\end{proof}

\section{Arithmetical Lemma}\label{sec:appendix_arithmetic}

We now prove Lemma~{\arithmeticalLemmaRef}, which motivates the definition of strongly violating the SCW condition.

\RestateLemmaNumber{\arithmeticalLemmaRef}
\begin{lemma}
\arithmeticalLemma
In addition, every such $q_j$ satisfies
\begin{equation*}
    q_j \;<\; p \;+\; 2 \left(\frac{n}{\ell + 1} - p\right) + 1.
\end{equation*}%
\end{lemma}

\begin{proof}
Define
\[
    k \;=\; p - n + \ell \left\lceil \frac{n - p}{\ell} \right\rceil.
\]
Using the bounds $x \leq \lceil x \rceil < x+1$, we obtain $0 \leq k < \ell$,
with equality $k = 0$ if and only if $\ell$ divides $n - p$.
The condition $q_1 + \cdots + q_{\ell} + p = n$ is then satisfied by setting
\begin{align*}
    q_1 = \cdots = q_k &= \left\lceil \frac{n - p}{\ell} \right\rceil - 1,
    \\
    q_{k+1} = \cdots = q_{\ell} &= \left\lceil \frac{n - p}{\ell} \right\rceil.
\end{align*}

Let $j \in [\ell]$. We now derive lower and upper bounds on~$q_j$.

For the lower bound, the assumption on $p$ implies
\[
    \frac{n - p}{\ell} - 1
    \;\geq\;
    \frac{n+1}{\ell+1} - 1
    \;=\;
    \frac{n - \ell}{\ell+1}
    \;\geq\; p.
\]
If $\ell$ divides $n-p$, then $k=0$ and
\[
    q_j \;=\; \frac{n - p}{\ell} \;>\; p.
\]
Otherwise,
\[
    q_j \;\geq\; \left\lceil \frac{n - p}{\ell} \right\rceil - 1
    \;>\; \frac{n - p}{\ell} - 1
    \;\geq\; p.
\]
Hence, in all cases, $q_j > p$.

For the upper bound, observe that
\[
    \frac{n - p}{\ell}
    \;=\;
    p + \Bigl(\frac{n}{\ell+1} - p\Bigr)\,\frac{\ell+1}{\ell},
\]
and therefore
\[
    q_j
    \;\leq\;
    \left\lceil \frac{n - p}{\ell} \right\rceil
    \;<\;
    p + 2\left(\frac{n}{\ell+1} - p\right) + 1,
\]
which completes the proof.
\end{proof}
}

\end{document}